\documentclass[11pt]{amsart}
\usepackage[thm,stdsize]{botex}
\usepackage{subfig}
\usepackage{cite}
\usepackage{amsthm,amsmath}
\usepackage{hyperref}

\usepackage{amsfonts}       
\usepackage{nicefrac}       
\usepackage{microtype}      
\usepackage{graphicx}
\usepackage{xcolor}

\newcommand{\conge}{\ensuremath{d}}

\newcommand{\diag}[1]{\mathrm{diag}(#1)}
\DeclareMathOperator{\dom}{dom}
\DeclareMathOperator{\supp}{supp}

\newcommand{\nonneg}{\ensuremath{\Real^m_{\ge 0}}}
\newcommand{\pos}{\ensuremath{\Real^m_{>0}}}

\newcommand{\ones}{\mathbf{1}}
\newcommand{\ftilde}{f_0}
\newcommand{\sinit}{s'}

\newcommand{\lyap}{\mathcal{E}}

\newcommand{\lmin}{\lambda_{\mathrm{min}}}
\newcommand{\lmax}{\lambda_{\mathrm{max}}}
\newcommand{\omegaset}{\ensuremath{\Omega_{\delta}}}

\renewcommand{\vec}{}

\newcommand{\mynote}[1]{ }

\allowdisplaybreaks

\begin{document}

\title{A Laplacian Approach to $\ell_1$-Norm Minimization}

\author{Vincenzo Bonifaci}



\begin{abstract}
We propose a novel differentiable reformulation of the linearly-constrained $\ell_1$ minimization problem, also known as the basis pursuit problem. The reformulation is inspired by the Laplacian paradigm of network theory and leads to a new family of gradient-based methods for the solution of $\ell_1$ minimization problems. We analyze the iteration complexity of a natural solution approach to the reformulation, based on a multiplicative weights update scheme, as well as the iteration complexity of an accelerated gradient scheme. The results can be seen as bounds on the complexity of iteratively reweighted least squares (IRLS) type methods of basis pursuit. 
\end{abstract}

\maketitle

\section{Introduction}
\label{sec:intro}
An important primitive in the areas of signal processing and optimization is that of finding a minimum $\ell_1$-norm solution to an underdetermined system of linear equations. Specifically, for some $n\le m$, let $\hat{s} \in \Real^m$ represent an unknown signal, $b \in \Real^n$ a measurement vector, and $A \in \Real^{n \times m}$ a full-rank matrix such that $A\hat{s}=b$. In some circumstances, the unknown signal $\hat{s}$ can be recovered by computing a minimum $\ell_1$-norm solution to the system $As=b$; in other words, solving the following optimization problem: 
\begin{align*}
\text{minimize } & \norm[1]{\vec s} \tag{BP} \\
\text{subject to } & \vec A \vec s = \vec b, \quad \vec s \in \Real^m. 
\end{align*}
This $\ell_1$-minimization problem is known as \emph{basis pursuit}. 
It is a central problem in the theory of sparse representation and arises in several applications, such as imaging and face recognition. Through a standard reduction, it also captures the \emph{$\ell_1$-regression} problem used in statistical estimation. 

The convex optimization problem (BP) can be cast as a linear program and thus could be solved via an interior-point method. Another popular approach to $\ell_1$-minimization  is the \emph{iteratively reweighted least squares} (IRLS) method, which is based on iteratively solving a series of adaptively weighted $\ell_2$-minimization problems. IRLS methods are popular in practice, due to their simplicity and the fact that they do not require  preprocessing nor special initialization rules. Despite this, theoretical guarantees for IRLS methods in the literature are not common, particularly in terms of global convergence bounds. 

This work contributes to developing the understanding and design of IRLS-type methods for basis pursuit. We propose a novel exact reformulation of (BP) as a differentiable convex problem over the positive orthant, which we call the \emph{dissipation minimization} problem. A distinguishing feature of this approach is that it entails the solution of a single differentiable convex problem. The reformulation leads naturally to a new family of IRLS-type methods solving (BP). 

We exemplify this approach by providing global convergence bounds for discrete IRLS-type algorithms for (BP). We explore two possible routes to the solution of the dissipation minimization problem, and thus of (BP), where we use the established framework of first-order optimization methods to derive two provably convergent iterative algorithms. 
We bound their iteration complexity as $O(m^2/\eps^3)$ and $O(m^2 / \eps^2)$, respectively, where $\eps$ is the relative error parameter. These methods are in the IRLS family since each iteration can be reduced to the solution of a weighted least squares problem. Both methods are very simple to implement and the first one exhibits a geometric convergence rate in numerical experiments. 

Our dissipation-based reformulation of (BP) may be of independent interest. It is rooted in the \emph{Laplacian framework} of network theory: it generalizes concepts such as the Laplacian matrix and the transfer matrix, which were originally developed to express the relation between electrical quantities across different terminals of a resistive network. (Many of our formulas have simple interpretations when the constraint matrix $A$ is derived from a network matrix). 

This paper is organized as follows. In Section \ref{sec:dmp}, we present the dissipation minimization reformulation of basis pursuit and some of its structural properties. In Section \ref{sec:equiv} we prove the equivalence between basis pursuit and dissipation minimization. In Section \ref{sec:ode} 
we look at the continuous dynamics obtained by applying mirror descent to the dissipation minimization objective and connect them with existing literature. 
In Section \ref{sec:algo}, we analyze a discretization of these dynamics that yields an iterative IRLS-type method for the solution of the dissipation minimization problem and, hence, of basis pursuit; this method can be seen as an application of the well-known multiplicative weights update scheme, and its iteration complexity is $O(m^2/\eps^3)$. Then, by leveraging Nesterov's accelerated gradient scheme, we present and analyze an improved IRLS-type method with iteration complexity $O(m^2/\eps^2)$. In Section \ref{sec:numerical}, implementations of the two methods are compared against existing solvers from the {\tt l1benchmark} suite \cite{Yang:2013}. 

\emph{Related literature.}
Given its central role in the areas of sparse representation and statistics, the literature on the basis pursuit problem and $\ell_1$-regression is extensive; see for example \cite{Bloomfield:1983,Candes:2005,Chen:2001,Foucart:2013} and references therein. Several algorithms for basis pursuit are reviewed in Chapter 15 of \cite{Foucart:2013}; for an experimental comparison and an application to face recognition, see \cite{Yang:2013}. 


Various versions of IRLS schemes have been studied for a long time \cite{Green:1984,Osborne:1985} and, as already mentioned, the methods have been popular in practice due to their simplicity and experimental performance \cite{Chartrand:2008}. On the other hand, theoretical guarantees for IRLS-type algorithms are few and far between \cite{Beck:2015,Daubechies:2010,Straszak:2016:irls}. 
A recent IRLS algorithm stands out in the context of this paper, as it applies to the basis pursuit problem and comes with a worst-case guarantee: a $\tilde{O}(m^{1/3} \eps^{-8/3})$ iterations algorithm due to Chin et al.~\cite[Theorem 5.1]{Chin:2013}, derived by further developing the approach of Christiano et al.~\cite{Christiano:2011}. 
In this context, our approach breaks the $\eps^{-8/3}$ bound for an IRLS method (at the cost of a worse dependency on $m$). We nevertheless emphasize that the goal of this work is  not to establish the superiority of a specific algorithm, but rather to highlight a new approach that, already when coupled with off-the-shelf optimization methods, offers a principled way to derive IRLS-type algorithms with competitive theoretical performance. 
Subsequently to the first appearance of our results (on arXiv), an improved bound of $\tilde{O}(m^{1/3} \eps^{-2/3} + \eps^{-2})$ iterations for a more sophisticated IRLS-type algorithm for (BP) has been derived by Ene and Vladu \cite{Ene:2019} (again building on the ideas of \cite{Christiano:2011} and \cite{Chin:2013}). While this algorithm has a rather more favorable worst-case dependency on the parameters, in practice it requires roughly $1/\eps$ iterations \cite[Section 4]{Ene:2019}; in contrast, as we observe in Section \ref{sec:numerical}, the experimental convergence rate of our approach is  \emph{geometric}, that is, the iterations required are linear in $\log (1/\eps)$, suggesting that a much stronger theoretical bound may hold in our setting. 

Our reformulation of basis pursuit is new, though it is in part inspired by the Laplacian framework \cite{Strang:1988}. In particular, the definition of the dissipation function is based on a generalization of the Laplacian potential of a network. This reinforces the idea from Chin et al.~\cite{Chin:2013} that concepts originally developed for network optimization can be fruitful in the context of $\ell_1$-regression. 
The dissipation-minimizing dynamics considered in Section \ref{sec:ode} are an application of the \emph{mirror descent} (or \emph{natural gradient}) dynamics \cite{Nemirovski:1983,Amari:2016,Arora:2012,Lu:2018} to our new objective function. 
In Section \ref{sec:pgs} we show, in particular, how the algorithmic framework of Lu, Freund and Nesterov \cite{Lu:2018} (see also \cite{Bauschke:2017}) can be applied to the dissipation minimization problem. 
The improved algorithm discussed in Section \ref{sec:ags} is instead based on Nesterov's well-known \emph{accelerated gradient method} \cite{Nesterov:2005}. 

The dynamics studied in Sections \ref{sec:ode} and \ref{sec:algo} bear some formal similarity to the so-called \emph{Physarum dynamics}, studied in the context of natural computing, which are the network dynamics of a slime mold \cite{Tero:2007,Bonifaci:2012,Straszak:2016:soda,Straszak:2016:irls,Becker:2017}. 
The fact that Physarum dynamics are of IRLS type was first observed in \cite{Straszak:2016:irls}. 
In this context, our result can be seen as the derivation of a Physarum-like dynamics purely from an optimization principle: dissipation minimization following the natural gradient. 
A relevant difference is that the specific dynamics we study is a gradient system, while the dynamics studied in \cite{Straszak:2016:irls,Becker:2017} is provably not a gradient system. This is precisely what enables us to apply the machinery of first-order convex optimization methods, and acceleration in particular. 

We note that a different proof of Theorem \ref{thm:equiv} has been independently provided by Facca, Cardin and Putti \cite{Facca:2019} in the context of the Physarum dynamics.

\emph{Notation.}
For a vector $x \in \Real^m$, we use $\diag{x}$ to denote the $m \times m$ diagonal matrix with the coefficients of $x$ along the diagonal. The inner product of two vectors $x, y \in \Real^m$ is denoted by $\innprod{x}{y} = x\tp y$. The maximum (respectively, minimum) eigenvalue of a diagonalizable matrix $M$ is denoted by $\lmax(M)$ (respectively, $\lmin(M)$). For a vector $x \in \Real^m$, $\norm[p]{x}$ denotes the $\ell_p$-norm of $x$ ($1 \le p \le \infty$), and $\abs{x}$ denotes the vector $y$ such that $y_i=\abs{x_i}$, $i=1,\ldots,m$. Similarly, $x^2$ denotes the vector $y$ such that $y_i=x_i^2$, $i=1,\ldots,m$. With a slight overlap of notation, which should nevertheless not cause any confusion, we instead reserve $x^k$ with a symbolic index $k$ to denote the vector produced by the $k$th step of an iterative algorithm.

\begin{table}
\begin{center}
\begin{tabular}{l|l}
\textbf{Algorithm} & \textbf{Iteration complexity} \\
\hline
Ref.~\cite{Chin:2013} & $\tilde{O}(m^{1/3} \eps^{-8/3})$ \\
PGS -- Theorem \ref{thm:pgs-for-dm} & $\tilde{O}(m^2 \eps^{-3})$ \\
AGS -- Theorem \ref{thm:ags-for-dm} & $O(m^2 \eps^{-2})$ \\
Ref.~\cite{Ene:2019} & $\tilde{O}(m^{1/3} \eps^{-2/3} + \eps^{-2})$ \\
\end{tabular}
\bigskip
\caption{Worst-case iteration complexity of recent IRLS methods for $\ell_1$-norm minimization}
\end{center}
\end{table}

\section{Basis pursuit and the dissipation minimization problem}
\label{sec:dmp}

\subsection{Assumptions on the basis pursuit problem}
\label{sec:bp}

We make the following assumptions on (BP): 
\begin{enumerate}
\item[(A.1)] the matrix $A$ has full rank and $n\le m$; 
\item[(A.2)] the system $As = b$ has at least one solution $\sinit$ such that $\sinit_j \neq 0$ for each $j=1,\ldots,m$. 
\end{enumerate}

\begin{proposition}
\label{prop:bp-assumption}
Assumption (A.2) is without loss of generality, given (A.1).
\end{proposition}
\begin{proof}
If the basis pursuit instance $(A,b)$ satisfies (A.1) but not (A.2), form a new instance $(A',b)$ where $A'$ is obtained from $A$ by duplicating every column. Observe the following about the two instances: 

\begin{itemize}
\item
$A'$ has full rank and $n' = n \le m \le 2m = m'$. 

\item
For any solution to $(A, b)$, there is a solution to $(A', b)$ with the same cost. 

\item 
Let $u = A\tp (A A\tp)^{-1} b$ be the least-square solution to $As=b$. There is at least one solution to $A' s' = b$ with $s'_j\neq 0$ for each $j=1,\ldots,2m$, given by 
\begin{align*}
s'_{2j-1} = \begin{cases} u_j/2 & \text{ if } u_j \neq 0, \\ +1 & \text{ if } u_j = 0, \end{cases}, \qquad 
s'_{2j} = \begin{cases} u_j/2 & \text{ if } u_j \neq 0, \\ -1 & \text{ if } u_j = 0. \end{cases}
\qquad j=1,\ldots,m. 
\end{align*}

\item 
No optimal solution to the instance $(A', b)$ is such that $s'_{2j-1} \cdot s'_{2j} < 0$ for some $j$: if that was the case, one could form a solution of lesser cost by replacing each of $s'_{2j-1}$ and $s'_{2j}$ with their average. Thus, any optimal solution $s'$ to $(A', b)$ can be transformed back into a solution $s$ to $(A, b)$ by taking $s_j = s'_{2j-1} + s'_{2j}$ for each $j=1,\ldots,m$. Such a solution satisfies $\norm[1]{s} = \norm[1]{s'}$ and thus must be optimal for $(A, b)$. 
\end{itemize}
\end{proof}

\begin{remark}
\label{ex:network}
A special case of (BP) is when $A$ is derived from a network matrix. Specifically, consider a connected network with $n+1$ nodes and $m$ edges, and suppose edge $j$ connects node $u$ to node $v$. Define $b_j \in \Real^m$ as $(b_j)_u = 1$, $(b_j)_v=-1$, and all other entries $0$. The matrix $B = [ b_1 \cdots b_m ] \in \Real^{(n+1) \times m}$ is called the \emph{incidence matrix} of the network. For any connected network, the incidence matrix $B$ has rank $n$ and, additionally, any row of $B$ can be expressed as a linear combination of the remaining $n$ rows, because the sum of all rows is a zero vector. Let $A$ be the submatrix of $B$ obtained by deleting an arbitrary row. Then $A$ satisfies assumption (A.1) and thus, without loss of generality, (A.2). A solution $s$ to $As = b$ can be interpreted as an assignment of flow values to each edge such that the net in-flow at every node $v=1,\ldots,n$ matches the prescribed demand $b_v$. 
\end{remark}

\subsection{The dissipation potential}
\label{sec:dm}

In this section we introduce the \emph{dissipation potential}, which is the function on which our reformulation of the basis pursuit problem is based. 

\begin{definition}
The \emph{Laplacian-like matrix} relative to a vector $\vec x \in \nonneg$ is the matrix $L(x) \defas A X A\tp$, where $X=\diag{x}$. 
\end{definition}

\begin{remark}
\label{rem:lapl}
In the network setting described in Remark \ref{ex:network}, 
a vector $x \in \pos$ can be interpreted as a set of weights, or \emph{conductances}, on the edges of the network. Then the matrix $B X B\tp$ is the \emph{weighted Laplacian} of the network \cite{Bollobas:1998,Strang:1988}. The matrix $L(x) = A X A\tp$ is sometimes called the \emph{reduced Laplacian}. 
\end{remark}

\begin{proposition}
\label{prop:pd-of-L}
If $x > 0$, then $L(x)$ is positive definite. 
\end{proposition}

\begin{proof}
Since $A$ has full rank, so has $A X^{1/2}$; hence $L(x) = (A X^{1/2}) (A X^{1/2})\tp$ is positive definite. 
\end{proof}

The following function definition is central to our approach. 
\begin{definition}
Let $A \in \Real^{n\times m}$, $b \in \Real^n$ be such that (A.1)--(A.2) hold. 
Define $\ftilde, f: \Real^m \to (-\infty,+\infty]$ as
\begin{align}
\ftilde(\vec x) &\defas 
\begin{cases} 
\ones\tp \vec x + \vec b\tp \vec L^{-1} (\vec x) \vec b, & \text{ if } \vec x \in \pos \\
+\infty & \text{ if } \vec x \notin \pos. \\
\end{cases} \\
\label{eq:liminf}
f(x) &\defas \liminf_{x' \to x} \ftilde(x'), \qquad\qquad   x \in \Real^m.
\end{align}
\end{definition}



We call $f$ the \emph{dissipation potential}. 
An equivalent definition of $f$ is as the \emph{convex closure} of $\ftilde$, which is the function whose epigraph in $\Real^{m+1}$ is the closure of the epigraph of $\ftilde$ \cite[Chapter~7]{Rockafellar:1970}. 
The \emph{effective domain} of $f$ is the set \[\dom f \defas \{x \in \Real^m \,:\, f(\vec x) < +\infty\}.\] 
The functions $f$ and $\ftilde$ differ only on the boundary of the positive orthant. 
We will show that $f$ always achieves a minimum on $\nonneg$, and hence on $\Real^m$. One of our main results (Theorem \ref{thm:equiv}) is that this minimum equals the minimum of (BP).

\mynote{%
Note that this coincides with the above $F$ in both the shortest-path setting, as well as for any $A$ with $n=1$.

The domain seems to be: 
\[
\{ 
\vec x \in \Real^m_{\ge 0}\, : \, \exists \vec s \in \Real^m \text{ s.t. } \vec A \vec s = \vec b \, \logand \, \supp(\vec s) \subseteq \supp(\vec x)
\}
\]
but establishing it formally is probably irrelevant. 
}%

\begin{remark}
Consider again the case where the matrix $A$ is derived from a network matrix, as in Remark \ref{ex:network}. The node of the network corresponding to the row that was removed from the incidence matrix to form $A$ is called the \emph{grounded} node. Now assume that for some $u = 1,\ldots,n$ the vector $b \in \Real^n$ is such that $b_v=0$ if $v \neq u$, $b_v = 1$ if $v = u$. Then the Laplacian potential $b\tp L^{-1}(x) b$ yields the \emph{effective resistance} between the grounded node and node $u$ when the conductances of the network are specified by the vector $x$. 
A standard result in network theory is that decreasing the conductance of any edge can only increase the effective resistance between any two nodes (see, for example, \cite{Bollobas:1998,Ghosh:2008}). 
Thus, the minimization of the dissipation potential $f$ involves an equilibrium between two opposing tendencies: decreasing any $x_j$ decreases the linear term $\ones\tp x$, but increases the Laplacian term $b\tp L^{-1}(x) b$. 
\end{remark}

\subsection{Basic properties of the dissipation potential}
We proceed to show that the dissipation potential attains a minimum. We start with some basic properties of $\ftilde$. 

\begin{lemma}
\label{thm:convexity}
The function $\ftilde$ is positive, convex and differentiable on $\pos$. 
\end{lemma}
\begin{proof}
Positivity follows from the positive-definiteness of $\vec L^{-1}(\vec x)$ for $\vec x \in \Real^m_{>0}$ (implied by Proposition \ref{prop:pd-of-L}). For convexity, it suffices to show that the mapping $\vec x \mapsto \vec b\tp L^{-1}(x) \vec b$ is convex on $\pos$. 
First observe that $\vec x \mapsto \vec A \vec X \vec A\tp$ is a linear matrix-valued function, i.e., each one of the entries of $\vec A \vec X \vec A\tp$ is a linear function of $\vec x$, since multiplying $\vec X$ on the left and right with $\vec A$ and $\vec A\tp$ yields linear combinations of the elements of $\vec x$. Second, the matrix to scalar function $\vec Y \mapsto \vec b\tp \vec Y^{-1} \vec b$ is convex on the cone of positive definite matrices, for any $\vec b \in \Real^n$ (see for example~\cite[Section 3.1.7]{Boyd:2004}). By combining the two facts above, it follows that the composition $\vec x \mapsto \vec b\tp (\vec A \vec X \vec A\tp)^{-1} \vec b$ is convex, and hence so is $\ftilde$. 
Finally, since the entries of $L(x)$ are linear functions of $x$, the function $\ftilde$ is a rational function with no poles in $\pos$, hence differentiable. 
\end{proof}

To argue that $f$ attains a minimum, we first recall some notions from convex analysis \cite{Rockafellar:1970,Beck:2017}. 
An extended real-valued function $f: \Real^m \to [-\infty,+\infty]$ is called \emph{proper} if its domain is nonempty and the function never attains the value $-\infty$. It is called \emph{closed} if its epigraph is closed. It is called \emph{coercive} if it is proper and $\lim_{\norm{x} \to \infty} f(x)=+\infty$. 

\begin{lemma}
\label{lem:closedness}
The function $f$ is nonnegative, proper, closed and convex on $\Real^m$. 
\end{lemma}
\begin{proof}
By Lemma \ref{thm:convexity}, $\ftilde$ is convex on $\Real^m$, since it is convex on its effective domain. Moreover $\ftilde$ is proper, since $L^{-1}(x)$ is positive definite  and thus $0<\ftilde(x)<+\infty$ for any $\vec x \in \pos$. By construction, $f$ coincides with the closure of $\ftilde$ and thus it is a closed proper convex function \cite[Theorem 7.4]{Rockafellar:1970}.  Its nonnegativity follows from the positivity of $\ftilde$ and from  \eqref{eq:liminf}. 
\end{proof}

\begin{corollary}
\label{cor:achieved}
The function $f$ attains a minimum on $\nonneg$. 
\end{corollary}
\begin{proof}
Note that $\lim_{\norm{\vec x} \to \infty} f(\vec x) = \infty$, because $\vec b\tp (\vec A \vec X \vec A\tp)^{-1} \vec b \ge 0$ for any $\vec x \in \dom \ftilde$, and $\ones\tp \vec x \to \infty$ as $\norm{\vec x} \to \infty$ with $x \in \dom \ftilde$. In other words, $f$ is also a \emph{coercive} function and therefore, it attains a minimal value over any nonempty closed set intersecting its domain \cite[Theorem 2.14]{Beck:2017}; in particular, it attains its minimal value over $\nonneg$.  
\end{proof}

Since $f(x) = \liminf_{x' \to x} \ftilde(x')$, the minimum attained by $f$ over $\nonneg$ equals $\inf_{x > 0} \ftilde(x)$. 
Note also that this minimum may be attained on the boundary of $\dom f$.

%
%

\subsection{Gradient and Hessian}
In this section we derive some formulas for the gradient and Hessian of $f$ on the interior of its domain. 


\begin{definition}
\label{def:conge}
Let $x \in \pos$. The \emph{voltage vector} at $x$ is  
$\conge(\vec x) \defas \vec A\tp \vec L^{-1}(x) \vec b \in \Real^m$. 
\end{definition}

\begin{remark}
In the network setting described in Remark \ref{ex:network}, $d_j(x)$ expresses the \emph{voltage along edge $j$} when an external current $b_u$ enters each node $u=1,\ldots,n$ (and a balancing current $-\sum_u b_u$ enters the grounded node). 
\end{remark}

The next lemma relates the gradient $\nabla f(x)$ to the voltage vector at $x$. 
\begin{lemma}
\label{lem:gradient}
Let $\vec x \in \pos$. For any $j=1,\ldots,m$, 
$ \frac{\partial f(x)}{\partial x_j} = 1 - (\vec a_j\tp \vec L^{-1}(\vec x) \vec b)^2 = 1 - \conge^2_j(\vec x),
$
where $a_j$ stands for the $j$th column of $\vec A$. 
\end{lemma}
\begin{proof}
First observe that 
$ \vec L(\vec x) = \vec A X \vec A\tp = \sum_{j=1}^m x_j \vec a_j \vec a_j\tp $
and thus $\partial \vec L/\partial x_j = a_j a_j\tp$. 
We apply the formula for the derivative of a matrix inverse: 
\begin{equation}
\label{eq:matrixinverse}
\frac{\partial L^{-1}}{\partial x_j} = - L^{-1} \frac{\partial L}{\partial x_j} L^{-1}. 
\end{equation} 
We obtain
\[
\frac{\partial \vec b\tp \vec L^{-1} \vec b}{\partial x_j} = -\vec b\tp \vec L^{-1} \frac{\partial \vec L}{\partial x_j} \vec L^{-1} \vec b = -\vec b\tp \vec L^{-1}  \vec a_j \vec a_j\tp \vec L^{-1} \vec b = - (\vec a_j\tp \vec L^{-1} \vec b)^2. 
\]
The claim follows by the definition of $f$. 
\end{proof}

\mynote{%
The gradient does not exist outside the relative interior of the domain [Rockafellar Corollary 25.1.1]. However, the subdifferential may not be empty there. In particular, it will not be empty at a global minimum (since the zero vector belongs to the subdifferential there). 
}%


To express the Hessian of $f$, in addition to the voltages we need the notion of transfer matrix. 
\begin{definition}
\label{def:T-matrix}
Let $x\in \pos$. 
The \emph{transfer matrix} at $x$ is $\vec T(x) \defas \vec A\tp \vec L^{-1}(x) \vec A.$ 
\end{definition}

\begin{remark}
In the network setting described in Remark \ref{ex:network}, the transfer matrix $T(x)$ expresses the relation between input currents and output voltages, when the conductances are given by the vector $x$. Namely, $T_{ij}(x)$ is the amount of voltage observed along edge $i$ of the network when a unit external current is applied between the endpoints of edge $j$. 
\end{remark}

\begin{corollary}
\label{cor:hessian}
For any $x > 0$, 
$
\nabla^2 f(\vec x) = 2 \cdot (\conge(\vec x) \cdot \conge(\vec x)\tp) \odot \vec T(\vec x), 
$
where $\odot$ denotes the \emph{Schur matrix product} defined by $(U \odot V)_{ij} = U_{ij} \cdot V_{ij}$. 
\end{corollary}
\begin{proof}
For any $i,j = 1,\ldots,m$, by Lemma \ref{lem:gradient} and applying once more  \eqref{eq:matrixinverse}, we get
\[ 
[\nabla^2 f(\vec x)]_{ij} = 2 (\vec b\tp \vec L^{-1} \vec a_i \vec a_i\tp \vec L^{-1} \vec a_j \vec a_j\tp \vec L^{-1} \vec b) = 2 \, \conge_i(\vec x) \conge_j(\vec x) \vec a_i\tp \vec L^{-1} \vec a_j. 
\]
The claim follows by Definition \ref{def:T-matrix}. 
\end{proof}


\mynote{%
Although the Hessian is always PSD, it can be singular at a global minimum even when this is unique. Consider $A=(1\; 1)$, $b=1$, $c_1=1/2$, $c_2=1$. Then $L^{-1}=1/(2x_1+x_2)$, $\conge_1=2/(2x_1+x_2)$, $\conge_2=1/(2x_1+x_2)$,  and the Hessian is
\[
\nabla^2 f(x) = \frac{2}{(2x_1+x_2)^3} \left(\begin{array}{cc}
4 & 2 \\
2 & 1 
\end{array} \right). 
\]
In other words, $f$ is convex but not strongly convex. 
}


\subsection{Bounds on the norms of gradient and Hessian}
\label{sec:norm-bounds}
In this section we derive some norm bounds for the gradient and Hessian of the dissipation potential $f$; they will be used crucially to derive complexity bounds for the algorithms studied in Section \ref{sec:algo}. 

Two matrices $M$, $M'$ are called \emph{congruent} if there is a nonsingular matrix $S$ such that $M' = SMS\tp$. 
For the proofs in this section, the main tool we rely on is the following algebraic fact relating the eigenvalues of congruent matrices; see for example \cite[Theorem 4.5.9]{Horn:2013} for a proof. 
\begin{theorem}[Ostrowski]
\label{thm:ostrowski}
Let $M, S \in \Real^{m\times m}$ be two symmetric matrices, with $S$ nonsingular. For $k=1,\ldots,m$, let $\lambda_k(M)$, $\lambda_k(SMS\tp)$ denote the $k$-th largest eigenvalue of $M$ and $SMS\tp$, respectively. For each $k=1,\ldots,m$ there is a positive real number $\theta_k \in [\lmin(S S\tp), \lmax(S S\tp)]$ such that 
\begin{equation}
\label{eq:ostrowski}
\lambda_k(S M S\tp) = \theta_k \lambda_k(M). 
\end{equation}
\end{theorem}

\begin{lemma}
\label{lem:spectrum-T}
Let $x \in \pos$. Each nonzero eigenvalue of $T(x)$ is at least $(\max_{i=1,\ldots,m} x_i)^{-1}$ and at most $(\min_{i=1,\ldots,m} x_i)^{-1}$. 
\end{lemma}
\begin{proof}
Consider the matrix $\Pi(x) \defas X^{1/2} T(x) X^{1/2}$. By Definition \ref{def:T-matrix}, 
\[
\Pi(x) = (A X^{1/2})\tp (A X A\tp)^{-1} (A X^{1/2}).
\]
Hence, $\Pi(x)$ is the orthogonal projection matrix that projects onto the range of $(A X^{1/2})\tp$. In particular, $\Pi(x)^2 = \Pi(x)$ and each eigenvalue of $\Pi(x)$ equals 0 or 1. Since $T(x) = X^{-1/2} \Pi(x) X^{-1/2}$, the matrices $T(x)$ and $\Pi(x)$ are congruent. By Theorem \ref{thm:ostrowski}, the algebraic multiplicity of the zero eigenvalue of $T(x)$ and $\Pi(x)$ is the same, and each positive eigenvalue of $T(x)$ must lie between the smallest and the largest eigenvalue of $X^{-1}$. These are $(\max_i x_i)^{-1}$ and $(\min_i x_i)^{-1}$, respectively. 
\end{proof}

\begin{lemma}
\label{lem:cong-bound}
Let $x \in \pos$. Then $\norm[\infty]{\conge(x)} \le (\min_{i=1,\ldots,m} x_i)^{-1} \cdot \norm[2]{s}$, 
where $s$ is any solution to $As=b$. In particular, for $c_{A,b} \defas b\tp (A A\tp)^{-1} b$, 
\begin{equation}
\label{eq:cong-bound}
\norm[\infty]{\conge(x)} \le (\min_{i=1,\ldots,m} x_i)^{-1} \, (c_{A,b})^{1/2}. 
\end{equation}
Additionally, if $s^*$ is an optimal solution to (BP), 
\begin{equation}
\label{eq:l1-vs-l2}
c_{A,b}^{1/2} \le \norm[1]{s^*} \le (m \cdot c_{A,b})^{1/2}. 
\end{equation}
\end{lemma}
\begin{proof}
Note that $\conge(x) = A\tp L^{-1}(x) b = A\tp L^{-1} A s = T(x) s$. Hence
\begin{equation}
\label{eq:txs}
	\norm[\infty]{\conge(x)} = \norm[\infty]{T(x) s} \le \norm[2]{T(x) s}. 
\end{equation}
Since the largest eigenvalue of $T(x)$ is at most $(\min_i x_i)^{-1}$ by Lemma \ref{lem:spectrum-T}, we can bound $\norm[2]{T(x) s} \le (\min_i x_i)^{-1} \norm[2]{s}$, proving the first part of the claim. For the second part, consider the least square solution $u \defas A\tp (A A\tp)^{-1} b$. Then $\norm[2]{u} = c_{A,b}^{1/2}$, and using the optimality of $u$ for the $\ell_2$ norm and of $s^*$ for the $\ell_1$ norm we derive
\[
c_{A,b} = \norm[2]{u}^2 \le \norm[2]{s^*}^2 \le \norm[1]{s^*}^2 \le \norm[1]{u}^2 \le m \norm[2]{u}^2 = m \cdot c_{A,b}. \qedhere
\]
\end{proof}



\begin{corollary}
\label{cor:grad-ub}
If $x \in \pos$, then
\begin{equation}
\norm[\infty]{\nabla f(x)} \le 1 + (\min_{i=1,\ldots,m} x_i)^{-2} \, c_{A,b}. 
\end{equation}
\end{corollary}
\begin{proof}
Combine Lemma \ref{lem:cong-bound} with Lemma \ref{lem:gradient}. 
\end{proof}

\begin{lemma}
\label{lem:hessian-ub}
If $x \in \pos$, then the largest eigenvalue of $\nabla^2 f(x)$ satisfies
\begin{equation}
\lmax(\nabla^2 f(x)) \le 2 \, (\min_{i=1,\ldots,m} x_i)^{-3} \cdot c_{A,b}.
\end{equation}
\end{lemma}
\begin{proof}
We can use the matrix identity $M \odot (z z\tp) = \diag{z} \cdot M \cdot \diag{z}$ to reexpress Corollary \ref{cor:hessian} as
\[
\nabla^2 f(x) = 2 D(x) T(x) D(x), 
\]
where $D(x) \defas \diag{\conge(x)}$. 
Hence, by Theorem \ref{thm:ostrowski}, the largest eigenvalue of $\nabla^2 f(x)$ satisfies
\begin{equation}
\label{eq:ostrowski-f}
\lmax(\nabla^2 f(x)) = 2 \, \theta \, \lmax(T(x))
\end{equation}
for some $\theta$ lying between the smallest and largest eigenvalues of $D(x)^2$. Since by Lemma \ref{lem:cong-bound}
\begin{equation}
\label{eq:bound-D}
\theta \le \lmax(D(x)^2) = \norm[\infty]{\conge(x)}^2 \le (\min_i x_i)^{-2} c_{A,b},
\end{equation}
combining \eqref{eq:ostrowski-f} and \eqref{eq:bound-D} with Lemma \ref{lem:spectrum-T} we get $\lmax(\nabla^2 f(x)) \le 2 (\min_i x_i)^{-3} c_{A,b}$. 
\end{proof}


\section{Equivalence between basis pursuit and dissipation minimization}
\label{sec:equiv}
In this section we prove the equivalence between basis pursuit and dissipation minimization. 
\begin{theorem}
\label{thm:equiv}
The value of the optimization problem
\begin{align*}
\text{\em minimize } & \norm[1]{\vec s} \tag{BP} \\
\text{\em subject to } & \vec A \vec s = \vec b, \quad \vec s \in \Real^m. 
\end{align*}
is equal to the value of the optimization problem
\begin{align*}
\text{\em minimize } & \frac{1}{2} \ones\tp x + \frac{1}{2} b\tp (AXA\tp)^{-1} b \tag{DM} \\
\text{\em subject to } & \vec x \in \pos. 
\end{align*}
\end{theorem}
We call (DM) the \emph{dissipation minimization problem} associated to $A$ and $b$. 
Note that the objective in (DM) is exactly $\ftilde(x)/2$, hence by \eqref{eq:liminf} the minimum of (DM) equals the minimum of $f(x)/2$ over $\nonneg$; the fact that this minimum is achieved is guaranteed by Corollary \ref{cor:achieved}.

\begin{definition}
Let $x > 0$. The \emph{solution induced by} $x$ is the vector $q(x) \defas X A\tp L^{-1}(x) b$. 
\end{definition}
The term ``solution'' is justified by the fact that $A q(x) = L L^{-1} b = b$. Induced solutions have the following simple characterization.

\mynote{%
The following lemma gives an alternative expression for $f$. 
\begin{lemma}
\[
f(\vec x) = \innprod{\vec x}{\vec x}_{\vec X^{-1}} + \innprod{\vec q(x)}{\vec q(x)}_{\vec X^{-1}}.
\]
\end{lemma}
\begin{proof}
Since $\vec A \vec q = \vec b$, and $q = X A\tp L^{-1} b$, we have $
\vec b\tp \vec L^{-1}(\vec x) \vec b = (\vec q\tp \vec A\tp) \vec L^{-1}(\vec x) b = q\tp X^{-1} q. $
\end{proof}
}%

\begin{lemma}
\label{lem:thomson}
Let $x \in \pos$. The solution induced by $x$, $q(x)$, equals the unique optimal solution to the quadratic optimization problem:   
\begin{align}
\label{eq:thomson}
\tag{QP$_x$} \text{\em minimize } & \vec s\tp \vec X^{-1} \vec s \\
\notag \text{\em subject to } & \vec A \vec s = \vec b, \quad s \in \Real^m. 
\end{align}
\end{lemma}
\begin{proof}
This lemma is a straightforward generalization of \emph{Thomson's principle} \cite[Chapter 9]{Bollobas:1998} from electrical network theory. 
We adapt an existing proof \cite[Lemma 3]{Bonifaci:2017:b} to the notation used in this paper. 
Since the objective function in \eqref{eq:thomson} is strictly convex, the problem has a unique optimal solution. 
Consider any solution $\vec s$, and let $\vec r= \vec s - \vec q(x)$. Then $\vec A \vec r = \vec b - \vec b = \vec 0$ and hence
$$
\vec s\tp \vec X^{-1} \vec s = (\vec q + \vec r)\tp \vec X^{-1} (\vec q + \vec r) =
\vec q\tp \vec X^{-1} \vec q + 2 \vec r\tp \vec X^{-1} \vec q + \vec r\tp \vec X^{-1} \vec r 
\ge 
\vec q\tp \vec X^{-1} \vec q,  
$$
since $\vec r\tp \vec X^{-1} \vec r \ge 0$ and $\vec r\tp \vec X^{-1} \vec q = \vec r\tp \vec A\tp \vec L^{-1} b = (\vec A \vec r)\tp \vec L^{-1} b = 0$. Therefore, the objective function value of any solution $\vec s$ to (QP$_x$) is at least as large as the objective function value of the solution $\vec q(x)$. 
\end{proof}

The value of (QP$_x$) is, in fact, the Laplacian potential $b\tp L^{-1}(x) b$. 
\begin{corollary}
\label{cor:qXq}
The minimum of (QP$_x$) equals $q(x)\tp X^{-1} q(x) = b\tp L^{-1}(x) b$. 
\end{corollary}
\begin{proof}
We already proved that the minimum of \textrm{(QP$_x$)} is $q(x)\tp X^{-1} q(x)$. Substituting the definition of $q(x)$, 
\[
q\tp X^{-1} q = (b\tp L^{-1} A\tp X) X^{-1} (X A\tp L^{-1} b) = b\tp L^{-1} L L^{-1} b = b\tp L^{-1} b. \qedhere
\]
\end{proof}

\begin{lemma}
\label{lem:bp-leq-dm}
For any $x > 0$, $q(x) \in \Real^m$ is such that $Aq=b$ and $\norm[1]{q(x)} \le f(x)/2$. Thus, the value of (BP) is at most that of (DM). 
\end{lemma}
\begin{proof}
For any $\vec x \in \pos$, consider its induced solution $q(x) = \vec X \vec A\tp \vec L(\vec x)^{-1} \vec b$. 
We already observed that $q(x)$ is feasible for (BP). Moreover, we can bound:  
\begin{align*}
\norm[1]{\vec q(x)} &= \vec x\tp \vec X^{-1} \abs{\vec q}  \\
& \le (\vec x\tp \vec X^{-1} \vec x)^{1/2} \, \cdot  (\vec q\tp \vec X^{-1} \vec q)^{1/2} & \\ 
& = (\ones\tp \vec x)^{1/2} \,\cdot (\vec b\tp \vec L^{-1}(\vec x) \vec b)^{1/2}  & \text{ (by Corollary \ref{cor:qXq})} \\
&\le \frac{1}{2} \ones\tp \vec x + \frac{1}{2} \vec b\tp L^{-1}(x) \vec b & \\ 
&= \frac{1}{2} f(\vec x), 
\end{align*}
where the first upper bound follows from the Cauchy-Schwarz inequality, and the second from the Arithmetic Mean-Geometric Mean inequality. 
\end{proof}

To prove the converse of Lemma \ref{lem:bp-leq-dm}, we develop an intermediate lemma that relates the value of an optimal solution $s^*$ of (BP) to the dissipation value of a vector $x$ such that $x=\abs{s}$ with $s$ sufficiently close to $s^*$. 
\begin{lemma}
\label{lem:vicinity}
Let $s \in \Real^m$, $\eps \in (0,1)$ be such that $\vec A \vec s = \vec b$, $s_j \neq 0$ and $(1-\eps) \abs{s^*_j} \le \abs{s_j} \le \abs{s^*_j} + \eps/m$ for some $s^*$ such that $A s^* = b$ and each $j=1,\ldots,m$. Then for $x=|s|$, 
\begin{equation}
\label{eq:vicinity-bound}
\frac{1}{2} f(x) \le \frac{\eps}{2} + \frac{1}{2} \left(1+\frac{1}{1-\eps} \right) \norm[1]{\vec s^*}. 
\end{equation}
\end{lemma}
\begin{proof}
On one hand, by the assumed upper bound $\abs{s_j} \le |s^*_j| + \eps/m$, trivially
\begin{equation}
\label{eq:f-norm-bound}
\ones\tp x = 
\norm[1]{\vec s} \le \norm[1]{\vec s^*} + \eps. 
\end{equation}

On the other hand, consider the solution $q(x)$ induced by $x$ and recall that $\vec q(\vec x)$ is feasible for (BP), since $\vec A \vec q = \vec b$, and optimal for \eqref{eq:thomson}. By the assumed lower bound $\abs{\vec s_j} \ge (1-\eps) \abs{s^*_j}$, and by Lemma \ref{lem:thomson}, 
\begin{align}
\vec b\tp \vec L^{-1}(\vec x) \vec b & = \vec q\tp \vec X^{-1} \vec q  \label{eq:lapl-bound}\\
& \le \vec s^{*\top} \vec X^{-1} \vec s^* = \sum_{j=1}^m \frac{1}{\abs{s_j}} (s_j^*)^2 \notag\\
& \le (1-\eps)^{-1} \sum_j \abs{s^*_j} = (1-\eps)^{-1}\norm[1]{\vec s^*},  \notag
\end{align}
where the first upper bound follows from the fact that $s^*$ is a solution to \eqref{eq:thomson}, and the second follows from the hypothesis. 
Combining \eqref{eq:f-norm-bound} and \eqref{eq:lapl-bound}, we get 
\[
\frac{1}{2} f(\vec x) \le \frac{1}{2} \norm[1]{\vec s^*} + \frac{\eps}{2} + \frac{1}{2} (1-\eps)^{-1} \norm[1]{\vec s^*}. \qedhere
\]
\end{proof}

\begin{lemma}
\label{lem:dm-leq-bp}
The value of (DM) is at most that of (BP). 
\end{lemma}
\begin{proof}
Consider an optimal solution $\vec s^* \in \Real^m$ to (BP).
%
%
Let $\sinit \in \Real^m$ be a solution to $\vec A \vec s = \vec b$ such that $\sinit_j \neq 0$ for all $j=1,\ldots,m$ (such an $\sinit$ exists by assumption (A.2)). 
For any $\delta \in (0,1)$, let $\vec s(\delta) \defas (1-\delta) \vec s^* + \delta \vec \sinit$ and $\vec x(\delta) \defas \abs{\vec s(\delta)} > 0$. For any $\eps \in (0,1)$ we can ensure that the hypotheses of Lemma \ref{lem:vicinity} are satisfied by choosing a small enough $\delta>0$. For such a value of $\delta$, Lemma \ref{lem:vicinity} yields
\begin{equation}
\label{eq:vicinity-bound-sopt}
\frac{1}{2} f(\vec x(\delta)) \le \frac{\eps}{2} + \frac{1}{2} \left(1+\frac{1}{1-\eps} \right) \norm[1]{\vec s^*}. 
\end{equation}
As $\eps$ can be chosen arbitrarily small, and the right-hand side of \eqref{eq:vicinity-bound-sopt} approaches $\norm[1]{\vec s^*}$ as $\eps \to 0$, we obtain the claim. 
\end{proof}


This concludes the proof of Theorem \ref{thm:equiv}. 
Not only are the optimal values of (BP) and (DM) the same, but one can bound the suboptimality of any feasible point of (BP) in terms of the dissipation value of a corresponding vector. 

\begin{theorem}
\label{thm:gap}
Let $s \in \Real^m$ be a feasible point of (BP) such that $s_j\neq 0$ for all $j=1,\ldots,m$, and let $x = \abs{s}$, $\rho(x) \defas \norm[\infty]{\conge(\vec x)}$. 
The quantity 
$ \left( 1 + \rho^{-1}(x) \right) \norm[1]{s} - \rho^{-1}(x) \cdot  f(x) $
is an upper bound on the suboptimality of $s$. 
\end{theorem}
\begin{proof}
Consider the following linear formulation of (BP) (left) and its dual (right):

\begin{minipage}[c]{0.45\textwidth}%
\begin{align*}
\text{minimize } & \ones\tp \vec x \\
\text{subject to } & \vec x + \vec s \ge 0 \\
& \vec x - \vec s \ge 0 \\
& \vec A \vec s = \vec b \\
& \vec x, \vec s \in \Real^m. 
\end{align*}
\end{minipage}
\begin{minipage}[c]{0.45\textwidth}%
\begin{align*}
\text{maximize } & \vec b\tp \vec \nu \\
\text{subject to } & \vec \lambda + \vec \mu = \ones \\
& \vec \lambda - \vec \mu + \vec A\tp \vec \nu = \vec 0 \\
& \vec \lambda, \vec \mu \ge \vec 0 \\
& \vec \lambda, \vec \mu \in \Real^m, \vec \nu \in \Real^n. 
\end{align*}
\end{minipage}

\medskip
Given any solution $s$ to (BP) such that $x = \abs{s}>0$, let us take
\begin{align*}
\nu &= \rho^{-1}(x) (\vec A \vec X \vec A\tp)^{-1} \vec b, \\
\lambda &= (\ones - \vec A\tp \vec \nu)/2, \\
\mu &= (\ones + \vec A\tp \vec \nu)/2. 
\end{align*}

Then $\norm[\infty]{A\tp \vec \nu} \le 1$ by definition of $\rho(x)$; moreover, $\lambda + \mu = \ones$, $\lambda - \mu + \vec A\tp \vec \nu = 0$, and $\lambda, \mu \ge 0$. Thus, $(\vec x, \vec s)$ is a primal feasible solution, $(\vec \lambda, \vec \mu, \vec \nu)$ is a dual feasible solution, and by weak duality
\[
\ones\tp \vec x \ge \vec b\tp \vec \nu = \rho^{-1}(x) \vec b\tp (\vec A \vec X \vec A\tp)^{-1} \vec b. 
\]
This implies a duality gap of 
\begin{align*}
\ones\tp x - \rho^{-1}(x) b\tp L^{-1}(x) b &= \ones\tp x - \rho^{-1}(x) (f(x) - \ones\tp x) \\
&= \left( 1+\rho^{-1}(x) \right) \norm[1]{s} - 2\rho^{-1}(x) \cdot \frac{1}{2} f(x). \qedhere
\end{align*}
\end{proof}

\mynote{%
As strong duality also holds, a corollary is that as $x \to x^*$, $\norm[\infty]{\conge(x)} \to \ones$. 
}


We close this section by observing that a simpler proof of Theorem \ref{thm:equiv} can be obtained by the following quadratic variational formulation of the $\ell_1$-norm: for any $s \in \Real^m$,
\[
\norm[1]{s} = \inf_{x\in \pos} \frac{1}{2} \sum_{j=1}^m \left( \frac{s_j^2}{x_j} + x_j \right),
\]
see, for example, Bach et al.~\cite[Section 1.4.2]{Bach:2012}. Therefore 
\begin{align*}
\min_{s \in \Real^m \atop As=b} \norm[1]{s} &= \min_{s \in \Real^m \atop As=b} \inf_{x \in \pos} \frac{1}{2} \left(  \frac{s_j^2}{x_j} + x_j \right) \\
&= \inf_{x \in \pos} \left( \frac{1}{2} \left( \min_{s \in \Real^m \atop As=b} s\tp X^{-1} s \right) + \frac{1}{2} \ones\tp x \right) \\
&= \inf_{x \in \pos} \left( \frac{1}{2} \, b\tp L^{-1}(x) b + \frac{1}{2} \, \ones\tp x \right),
\end{align*}
where the last identity follows from Corollary \ref{cor:qXq}. 
However, the full strength of Lemma \ref{lem:bp-leq-dm} and Lemma \ref{lem:dm-leq-bp} is crucial to be able to constructively transform feasible points for (DM) into feasible points for (BP) and vice versa. 

\section{Continuous dynamics for dissipation minimization}
\label{sec:ode}
Theorem \ref{thm:equiv} readily suggests an approach to the solution of the basis pursuit problem. 
Namely, the solution of the non-smooth, equality constrained formulation (BP) is reduced to the solution of the differentiable formulation (DM) on the positive orthant. 

\emph{Mirror descent dynamics.}
To solve (DM), it is natural to adopt methods for differentiable constrained optimization that are designed for simple constraints. 
Consider first the following set of ordinary differential equations, aimed at solving
$\inf \, \{ f(\vec x) \,|\, \vec x > 0\}$: 
\begin{equation}
\label{eq:ode-evol}
\dot{x}_j = -x_j \frac{\partial f(x)}{\partial x_j}, \qquad j=1,\ldots,m,
\end{equation}
with initial condition $x(0) = x^0$ for some $x^0 > 0$. 
When $f$ is the dissipation potential, by Lemma \ref{lem:gradient} this yields the explicit dynamics
\begin{equation}
\label{eq:explicit}
\dot{x}_j = x_j(\conge^2_j(x) - 1) = x_j( (a_j\tp (AXA\tp)^{-1} b)^2 - 1), \qquad j=1,\ldots,m. 
\end{equation}


The dynamical system \eqref{eq:ode-evol} is a nonlinear Lotka-Volterra type system of differential equations, of a kind that is common in population dynamics \cite{Hofbauer:1998}. It is also an example of a \emph{Hessian gradient flow} \cite{Alvarez:2004}: it can be expressed in the form
\begin{equation}
\label{eq:evol}
\dot{\vec x} = - \vec H^{-1}(\vec x) \nabla f(\vec x)
\end{equation}
where $\vec H(\vec x) = \nabla^2 h(\vec x)$ is the Hessian of a convex function $h$; namely, here $\vec H(\vec x) = \vec X^{-1}$, and $h: \pos \to \Real$ is the \emph{negative entropy} function
\begin{equation}
\label{eq:negentropy}
h(\vec x) \defas \sum_{j=1}^m x_j \ln x_j - \sum_{j=1}^m x_j. 
\end{equation}
System \eqref{eq:evol} can also be expressed as
$
\frac{d}{dt} \frac{\partial h(\vec x)}{\partial x_j} = -\frac{\partial f(\vec x)}{\partial x_j}, j=1,\ldots,m,
$
or more succinctly, 
\begin{equation}
\label{eq:mdf}
\frac{d}{dt} \nabla h(\vec x) = - \nabla f(\vec x), 
\end{equation}
which is known as the \emph{mirror descent dynamics} or \emph{natural gradient flow} \cite{Nemirovski:1983,Amari:2016}. 
The well-posedness of \eqref{eq:evol} has been considered, for example, in \cite{Alvarez:2004}. 
A dynamics formally similar to \eqref{eq:explicit} is the \emph{Physarum dynamics} \cite{Bonifaci:2012,Straszak:2016:soda,Straszak:2016:irls,Becker:2017}, namely, 
\begin{equation}
\label{eq:physarum}
\dot{x}_j = x_j(\abs{\conge_j(x)} - 1) = x_j( \abs{a_j\tp (AXA\tp)^{-1} b} - 1), \qquad j=1,\ldots,m.  
\end{equation}
Differently from \eqref{eq:explicit}, the dynamics \eqref{eq:physarum} is \emph{not} a gradient flow, that is, there is no function $f$ that allows to write the dynamics in the form \eqref{eq:evol} or \eqref{eq:mdf} (with $h$ the negative entropy). 

\emph{Convergence of the dynamics.}
The fact that the solution of the mirror descent dynamics \eqref{eq:evol} converges to a minimizer of $f$ with rate $1/t$ is a well-known result; see, for example, \cite{Alvarez:2004,Wilson:2018}. 
We include a streamlined proof for completeness. 


\begin{lemma}
\label{lem:f-decreasing}
The values $f(x(t))$ with $x(t)$ given by \eqref{eq:ode-evol} are nonincreasing in $t$. 
\end{lemma}
\begin{proof}
We compute
\[
\frac{d}{dt} f(\vec x(t)) = \sum_{j=1}^m \frac{\partial f}{\partial x_j}(\vec x) \frac{d x_j}{dt} (\vec x) = - \sum_{j=1}^m x_j \left( \frac{\partial f}{\partial x_j}(\vec x) \right)^2 \le 0.  \qedhere
\]
\end{proof}

A key role in the convergence of the mirror descent dynamics is played by the \emph{Bregman divergence} of the function $h$. 

\begin{definition}
The \emph{Bregman divergence} of a convex function $h: \Real^m \to (-\infty,+\infty]$ is defined by
$
D_h(x, y) \defas h(x) - h(y) - \innprod{\nabla h(y)}{x-y}. 
$
\end{definition}
Convexity of $h$ implies the nonnegativity of $D_h(x, y)$. When $h$ is the negative entropy, $D_h$ is the \emph{relative entropy} function (also known as Kullback-Leibler divergence), for which $D_h(x,y)=0$ if and only if $x=y$.

\begin{theorem}[$\!\!$\cite{Alvarez:2004,Wilson:2018}]
\label{thm:md-dynamics}
Let $x^* \in \nonneg$ be a minimizer of $f$. 
As $t \to \infty$, the values $f(x(t))$ with $x(t)$ given by \eqref{eq:ode-evol} converge
to $f(x^*)$. 
In particular, 
\[
f(x(t)) - f(x^*) \le \frac{1}{t} D_h(x^*, x(0)) = O\left(\frac{1}{t}\right). 
\]
\end{theorem}
\begin{proof}

In the following, to shorten notation we often write $x$ in place of $x(t)$. 
Since $(d/dt) \nabla h(x) + \nabla f(x) = 0$ by \eqref{eq:evol}, for any $y$ we have $\innprod{(d/dt) \nabla h(x) + \nabla f(x)}{x - y} = 0$. This is equivalent to 
\begin{equation}
\label{eq:conv1}
\innprod{\frac{d}{dt} \nabla h(x)}{ x - y} + \innprod{\nabla f(x)}{ x - y} = 0.
\end{equation}
On the other hand, since $(d/dt) h(x) = \innprod{\nabla h(x)}{\dot x}$, a simple calculation shows
\begin{equation}
\label{eq:conv2}
\frac{d}{dt} D_h(y, x) = \innprod{\frac{d}{dt} \nabla h(x)}{ x - y}. 
\end{equation}
Combining \eqref{eq:conv1} and \eqref{eq:conv2}, 
and plugging in $y=x^*$, 
\begin{equation}
\label{eq:d-bregman}
\frac{d}{dt} D_h(x^*, x) = -\innprod{\nabla f(x)}{x-x^*}. 
\end{equation}


The proof is concluded by a potential function argument \cite{Bansal:2019,Wilson:2018}. Consider the function
\[ \lyap(t) \defas D_h(x^*, x) + t(f(x) - f(x^*)). \]
Its time derivative is, by \eqref{eq:d-bregman}, 
\[
\frac{d}{dt} \lyap(t) = - \innprod{\nabla f(x)}{x-x^*} + f(x) - f(x^*) + t \frac{d}{dt} f(x), 
\]
where the last summand is nonpositive by Lemma \ref{lem:f-decreasing}  and the other terms equal, by definition, $-D_f(x^*, x) \le 0$. 
Hence, $\lyap(t) \le \lyap(0)$ for all $t \ge 0$, which is equivalent to
\[
D_h(x^*, x) + t (f(x) - f(x^*)) \le D_h(x^*, x(0)),
\]
proving the claim. 
\end{proof}

\section{Algorithms for dissipation minimization}
\label{sec:algo}

We now turn to the problem of designing IRLS-type algorithms for (DM) (and thus (BP)) with provably bounded iteration complexity. 
Two technical obstacles in the setup of a first-order method for formulation (DM) are: 1) that the positive orthant is not a closed set, and 2) that the gradients of $f$ may not be uniformly bounded on the positive orthant. There is a way to deal with both issues at once: instead of solving $\inf_{x>0} f(x)$, for an appropriately small $\delta>0$
one can minimize $f$ over 
\[ \omegaset \defas \{ x \in \Real^m \,:\, \delta \ones \le x  \}.\]
This is established by the next lemma. 
\begin{lemma}
\label{lem:cut-orthant}
Let $x^*$ be a minimizer of $f$. Then $f(x^*) \le \min_{x \in \omegaset} f(x) \le f(x^*) + \delta\, m$. 
\end{lemma}
\begin{proof}
The first inequality is trivial. As for the second, recall that $f(x) = \ones\tp x + b\tp L^{-1}(x) b$ for any $x > 0$, and that in the latter sum, the second term is non-increasing with $x$ (by Lemma \ref{lem:gradient}). Thus, for any $x > 0$, 
\[
f(x+\delta \ones) = \ones\tp(x+\delta \ones) + b\tp L^{-1}(x+\delta \ones) b \le \delta m + f(x). 
\]
In other words, for any $x>0$, there is $y \ge \delta \ones$ (namely, $y = x+\delta\ones$) such that $f(y) \le f(x)+\delta m$. 
\end{proof}

In the following, we let $\delta \defas \eps\, c_{A,b}^{1/2}/(2m)$, where $\eps$ is the desired error factor and $c_{A,b}$ is as defined in Lemma \ref{lem:cong-bound}; this, by Lemma \ref{lem:cong-bound} and Theorem \ref{thm:equiv}, ensures that the additional error incurred by restricting solutions to $\omegaset$ is at most $(\eps/2) \norm[1]{s^*} = (\eps/4) f(x^*)$. 



\subsection{Primal gradient scheme}
\label{sec:pgs}
Guided by \eqref{eq:mdf}, we might consider its forward Euler discretization
\begin{equation}
\label{eq:md}
\nabla h(x^{k+1}) - \nabla h(x^k) = - \eta \nabla f(x^k),
\end{equation}
where $x^k \in \omegaset$ denotes the $k$th iterate, and $\eta \in \Real_{>0}$ an appropriate step size. 
Indeed, the update \eqref{eq:md} falls within a well-studied methodology for first-order convex optimization \cite{Beck:2003,Lu:2018}. We adapt this framework to the solution of (DM).

The \emph{primal gradient scheme} is a first-order method for minimizing a differentiable convex function $f$ over a closed convex set $Q$. This scheme, which is defined with respect to a reference function $h$, proceeds as follows \cite{Lu:2018,Bauschke:2017}:
\begin{enumerate}
\item Initialize $x^0 \in Q$. Let $\beta > 0$ be a parameter. 
\item At iteration $k=0,1,\ldots$, compute $\nabla f(x^k)$ and set
\begin{equation}
\label{eq:pgs}
x^{k+1} \leftarrow \argmin_{x \in Q} \{ \innprod{\nabla f(x^k)}{x-x^k} + \beta D_h(x, x^k) \}. 
\end{equation}
\end{enumerate}

We apply the scheme with $h$ as defined in \eqref{eq:negentropy} and with $Q = \omegaset$. Then, the minimization in \eqref{eq:pgs} can be carried out analytically; it reduces to
\begin{equation}
\label{eq:explicit-md}
x^{k+1}_j = \max \{ \delta, x^k_j \cdot \exp(-\beta^{-1} [\nabla f(x^k)]_j) \}, \qquad j = 1,\ldots,m. 
\end{equation}
Update \eqref{eq:explicit-md} is straightforward to implement as long as one can compute $\nabla f(x^k)$. This computation is discussed in Section \ref{sec:one-iteration}.

\emph{Convergence of the primal gradient scheme.}
As shown in \cite{Lu:2018}, the primal gradient scheme achieves an absolute error bounded by $O(\beta/k)$ after $k$ iterations provided that the function $f$ is \emph{$\beta$-smooth relative to $h$}. In our case, where both $f$ and $h$ are twice-differentiable on $Q$, relative $\beta$-smoothness is defined as
\begin{equation}
\label{eq:smoothness-condition}
\lmax(\nabla^2 f(x)) \leq \beta \,\cdot \lmax(\nabla^2 h(x)) \qquad \text{ for all } x \in Q. 
\end{equation}

\begin{theorem}[$\!\!$\cite{Lu:2018}]
\label{thm:pgs}
If $f$ is $\beta$-smooth relative to $h$, then for all $k \ge 1$, the updates \eqref{eq:pgs} satisfy
\[
f(x^k) - f(x^*|_Q) \le \frac{\beta}{k} D_h(x^*|_Q, x^0). 
\] 
where $x^*|_Q \defas \argmin_{x \in Q} f(x)$. 
\end{theorem}

To apply Theorem \ref{thm:pgs} in our setting, we need to bound the smoothness parameter $\beta$. We do this by leveraging the bounds derived in Section \ref{sec:norm-bounds}. 
\begin{lemma}
\label{lem:beta-pgs}
Equation \eqref{eq:smoothness-condition} holds for $\beta = 8 m^2 /\eps^2 $. 
\end{lemma}
\begin{proof}
Condition \eqref{eq:smoothness-condition} is equivalent to the condition that the largest eigenvalue of the matrix 
\[ [\nabla^2 h(x)]^{-1} \nabla^2 f(x) = X \nabla^2 f(x) \]
 be at most $\beta$ (see \cite[Theorem 7.7.3]{Horn:2013}). The matrix $X \nabla^2 f(x)$ is similar to $X^{1/2} \nabla^2 f(x) X^{1/2}$, hence it suffices to bound the eigenvalues of the latter. Since $\nabla^2 f(x) = 2 D(x) T(x) D(x)$ with $D(x) = \diag{\conge(x)}$, 
 \[
 X^{1/2} \nabla^2 f(x) X^{1/2} = 2 X^{1/2} D T D X^{1/2} = 2 D X^{1/2} T X^{1/2} D = 2 D \Pi D, 
 \]
 where we used the fact that $X$ and $D(x)$ are diagonal. By the proof of Lemma \ref{lem:spectrum-T}, the eigenvalues of $\Pi(x)$ are all 0 or 1. Hence, using again the relation between the eigenvalues of congruent matrices (Theorem \ref{thm:ostrowski}), we conclude that the largest eigenvalue of $ X^{1/2} \nabla^2 f(x) X^{1/2}$ is bounded by that of $2 D(x)^2$. Since $D(x) = \diag{\conge(x)}$, the latter equals $2 \norm[\infty]{\conge(x)}^2$, which is $2 c_{A,b} / \delta^2 = 8 m^2/\eps^2$ by Lemma \ref{lem:cong-bound} and the definitions of $\omegaset$ and $\delta$. 
\end{proof}

\begin{theorem}
\label{thm:pgs-for-dm}
The primal gradient scheme \eqref{eq:explicit-md} applied to the dissipation minimization problem (DM) achieves relative error at most $\eps$ after $96 m^2 \log(m/\eps) / \eps^3 = \tilde{O}(m^2/\eps^3)$ iterations. 
\end{theorem}
\begin{proof}
By Theorem \ref{thm:pgs} and Lemma \ref{lem:beta-pgs}, after $k$ iterations it holds that
\begin{equation}
f(x^k) - f(x^*|_Q) \le 8 R m^2 / (k \eps^2),  
\end{equation}
where $R \defas D_h(x^*|_Q, x^0)$. 
Since $f(x^*|_Q) \le (1+\eps/4) f(x^*)$ (by Lemma \ref{lem:cut-orthant}, since $Q=\omegaset$), this implies
\begin{equation}
f(x^k) - f(x^*) \le 8 R m^2 / (k \eps^2) + \frac{\eps}{4} f(x^*). 
\end{equation}
Thus, $f(x^k) - f(x^*) \le \eps f(x^*)$ if we take $k = \lceil 32 R m^2/(3 \eps^3 f(x^*)) \rceil$. We complete the proof by bounding $R/f(x^*)$ in terms of $\log (m/\eps)$. Let 
\[
\mu \defas \max_{j=1,\ldots,m} \frac{x_j^*|_Q}{x_j^0}
\]
Observe that since $x^0 \in Q$, 
\[
\mu \le \frac{1}{\delta} \max_j  x_j^*|_Q \le \frac{2m}{\eps c_{A,b}^{1/2}} f(x^*|_Q) \le \frac{2m^{3/2}}{\eps (m c_{A,b})^{1/2}} (1+\eps/4) f(x^*) \le \frac{8m^{3/2}}{\eps f(x^*)} f(x^*)
\]
with the last inequality following from \eqref{eq:l1-vs-l2}. Thus, 
\begin{align*}
R &= \sum_{j=1}^m x_j^*|_Q \log \frac{x_j^*|_Q}{x_j^0} \le (\log \mu) \, f(x^*|_Q) \le (\log \mu)(1+\eps/4) f(x^*) \\
&\le 2 \log \left( \frac{8m^{3/2}}{\eps} \right) f(x^*) \le 9 \log\left( \frac{m}{\eps} \right) f(x^*). 
\end{align*}
Hence, $k = \lceil 96 m^2 \log(m/\eps) / \eps^3 \rceil$ iterations suffice to achieve relative error $\eps$. 
\end{proof}

\subsection{Accelerated gradient scheme}
\label{sec:ags}
The second optimization scheme that we consider is the \emph{accelerated gradient method} of Nesterov \cite{Nesterov:2005}. This can be summarized as follows:  

\begin{enumerate}
\item Initialize $x^0 \in Q$. Let $\beta > 0$ be a parameter. 
\item At iteration $k=0,1,\ldots$, compute $\nabla f(x^k)$ and set $\alpha_k=1/2(k+1)$, $\tau_k = 2/(k+3)$ and
\begin{align}
y^k &\leftarrow \argmin_{x \in Q} \left\{ \frac{\beta}{2} \norm{x-x^k}^2_2 + \innprod{\nabla f(x^k)}{x - x^k} \right\} \label{eq:agd-1}
\\
z^k &\leftarrow \argmin_{x \in Q} \left\{ \frac{\beta}{2} \norm{x-x^0}^2_2 + \sum_{i=0}^k \alpha_i \innprod{\nabla f(x^i)}{ x- x^i} \right\} \label{eq:agd-2}
\\
x^{k+1} &\leftarrow \tau_k z^k + (1-\tau_k) y^k. \label{eq:agd-3}
\end{align}
\end{enumerate}

In our application of the scheme, $Q=\omegaset$ and the minimization in \eqref{eq:agd-1} and \eqref{eq:agd-2} can be carried out analytically; explicitly, they become
\begin{align}
\label{eq:explicit-ags-1}
y^{k}_j &= \max \{ \delta, x^k_j - \beta^{-1} [\nabla f(x^k)]_j \}, & & j = 1,\ldots,m \\
\label{eq:explicit-ags-2}
z^{k}_j &= \max \{ \delta, x^0_j - \beta^{-1} [\sum_{i=0}^k \alpha_i \nabla f(x^i)]_j \}, & & j = 1,\ldots,m. 
\end{align}
To implement \eqref{eq:explicit-ags-1}--\eqref{eq:explicit-ags-2}, it is enough to be able to access the gradient $\nabla f(x^k)$ and the cumulative gradient $\sum_i \alpha_i \nabla f(x^i)$; the latter can be maintained with one additional update at each iteration. 

\emph{Convergence of the accelerated gradient scheme.}
The well-known result by Nesterov \cite{Nesterov:2005} shows that the accelerated gradient scheme achieves an absolute error bounded by $O(\beta/k^2)$ after $k$ iterations provided that the gradient of the function $f$ is \emph{$\beta$-Lipschitz-continuous} over $Q$. In our case, where $f$ is twice-differentiable on $Q$, this means
\begin{equation}
\label{eq:lipschitz-condition}
\lmax(\nabla^2 f(x)) \le \beta \qquad \text{ for all } x \in Q. 
\end{equation}

\begin{theorem}[$\!\!$\cite{Nesterov:2005}]
\label{thm:ags}
If $\nabla f$ is $\beta$-Lipschitz-continuous over $Q$, then for all $k\ge 1$, the updates \eqref{eq:agd-1}--\eqref{eq:agd-3} satisfy
\[
f(y^k) - f(x^*|_Q) \le \frac{2 \beta}{(k+1)^2} \norm[2]{x^*|_Q - x^0}^2
\]
where $x^*|_Q \defas \argmin_{x \in Q} f(x)$. 
\end{theorem}

Again, to apply Theorem \ref{thm:ags} in our setting, we need to bound the smoothness parameter $\beta$. We do this by exploiting Lemma \ref{lem:hessian-ub}. 
\begin{lemma}
\label{lem:beta-ags}
Equation \eqref{eq:lipschitz-condition} holds for $\beta = 16 m^3/(\eps^3 c_{A,b}^{1/2})$. 
\end{lemma}
\begin{proof}
Immediate from Lemma \ref{lem:hessian-ub}, the fact that $Q=\omegaset$ and the definition of $\omegaset$. Recall that $\delta = \eps c_{A,b}^{1/2}/(2m)$. 
\end{proof}

\begin{theorem}
\label{thm:ags-for-dm}
If $x^0=\abs{u}$ where $u\defas A\tp (A A\tp)^{-1} b$ is the least square solution to $As=b$, the accelerated gradient scheme \eqref{eq:agd-1}--\eqref{eq:agd-3} applied to the dissipation minimization problem (DM) achieves relative error at most $\eps$ after 
 $24 m^2/\eps^2$ iterations. 
\end{theorem}
\begin{proof}
By Theorem \ref{thm:ags} and Lemma \ref{lem:beta-ags}, after $k$ iterations it holds that
\begin{equation}
f(y^k) - f(x^*|_Q) \le 32  R  m^3/(k^2 \eps^3 c_{A,b}^{1/2}).
\end{equation}
Since $f(x^*|_Q) \le (1+\eps/4) f(x^*)$ by Lemma \ref{lem:cut-orthant}, this implies
\begin{equation}
f(y^k) - f(x^*) \le (\eps/4) f(x^*) + 32 R  m^3/(k^2 \eps^3 c_{A,b}^{1/2}). 
\end{equation}
Thus, $f(y^k) - f(x^*) \le (\eps/4) f(x^*) + (\eps/2) c_{A,b}^{1/2} < \eps f(x^*)$ if the number of iterations $k$ is at least
\begin{equation}
\label{eq:intermediate-bound-ags}
\frac{8 m^{3/2}}{\eps^2} \left( \frac{R}{c_{A,b}} \right)^{1/2}. 
\end{equation}
We complete the proof by bounding $(R/c_{A,b})^{1/2} = \norm[2]{x^*|_Q - x^0} / c_{A,b}^{1/2}$ in terms of $m$. Observe that $R^{1/2} = \norm[2]{x^*|_Q - x^0} \le \norm[2]{x^*|_Q} + \norm[2]{x^0}$. By the assumption that $x^0=|u|$ where $u$ is the least square solution to $As=b$, $\norm[2]{x^0} = \norm[2]{u} = c_{A,b}^{1/2}$ (recall the definition of $c_{A,b}$ in Lemma \ref{lem:cong-bound}). Moreover,
\begin{align*}
\norm[2]{x^*|_Q} \le \norm[1]{x^*|_Q} \le \frac{1}{2} f(x^*|_Q) \le \frac{1}{2} f(x^*) + \frac{\eps}{2} c_{A,b}^{1/2} \le (m \, c_{A,b})^{1/2} + c_{A,b}^{1/2} < 2 m^{1/2} c_{A,b}^{1/2}.
\end{align*}
Hence $(R/c_{A,b})^{1/2} \le 3m^{1/2}$ and substitution in \eqref{eq:intermediate-bound-ags} yields the theorem. 
\end{proof}

\subsection{Implementing the iterations}
\label{sec:one-iteration}

We conclude this section by commenting on a few implementations details and in particular on how each iteration of \eqref{eq:explicit-md} and \eqref{eq:agd-1}--\eqref{eq:agd-3} could be implemented. A notable point is that each iteration can be reduced to a series of operations that access the matrix $A$ only through the solution of a system of the form $A W A\tp p = b$, for some diagonal matrix $W$, or through matrix-vector multiplications of the form $Ax$ or $A\tp x$. 

\emph{Computation of the gradient.}
By Lemma \ref{lem:gradient}, computing the vector $\conge(x) = A\tp L^{-1}(x) b$ is enough to compute the gradient at $x$, since $\nabla f(x) = \ones-\conge^2(x)$. To compute $\conge(x)$, it is enough to solve the linear system $L(x) p = b$ for $p$, then premultiply the solution with $A\tp$. Note that since $L(x) = A X A\tp$, the system $L(x) p = b$ is a symmetric linear system with a positive definite constraint matrix. 

\emph{Warm start.} Heuristically, the solution of the system $L(x^{k+1}) p = b$, which is required to compute the gradient at iteration $k+1$, can be expected to be close to that of the system $L(x^k) p = b$ when $x^{k+1}$ is close to $x^k$. Hence, one possibility in practice is to use the solution obtained at step $k$ to warm-start the linear equation solver at step $k+1$, with a possible substantial reduction in the computational cost of each iteration. 

\emph{Initial point and exit criterion.} 
We assumed the starting point is the least square solution in Theorem \ref{thm:ags-for-dm}, but this was only to optimize the worst-case iteration bound. In fact, Theorem \ref{thm:pgs-for-dm} and Eq.~\eqref{eq:intermediate-bound-ags} always apply and the schemes we discussed do not require a special initialization apart from membership into $\omegaset$; hence, any point that is not too close to the boundary of the positive orthant is a suitable starting point. 
We can stop the schemes after the number of iterations $k$ is large enough to ensure the error guarantees of Theorems \ref{thm:pgs-for-dm} and \ref{thm:ags-for-dm} (or Eq.~\eqref{eq:intermediate-bound-ags}). 
Alternatively, a natural exit criterion in practice can be based on the duality gap provided by Theorem \ref{thm:gap}. 

\emph{Obtaining feasible iterates for (BP).}
The algorithms as described above produce iterates in the positive orthant, that is, iterates that are feasible for (DM), but after all, our goal was to obtain feasible iterates of (BP). By using the ideas of Lemma \ref{lem:bp-leq-dm}, we can easily associate with any iterate $x^k \in \pos$ an iterate $s^k$ that is feasible for (BP), and the cost of which is not larger than the dissipation cost of $x^k$: namely, take $s^k = q(x^k) = X^k A\tp L(x^k)^{-1} b$. By the proof of Lemma \ref{lem:bp-leq-dm}, we know that $\norm[1]{s^k} \le f(x^k)/2$. Thus, the error bounds for $f(x^k)$ can be directly translated into error bounds for $\norm[1]{s^k}$.  Note that $s^k$ can be computed essentially for free, since $s^k = X^k d(x^k)$ and $d(x^k)$ is a byproduct of the gradient computation at iteration $k$. 


\section{Numerical comparison with other algorithms for $\ell_1$-minimization}
\label{sec:numerical}
We include in this section a numerical comparison of our schemes to other well-known algorithms for $\ell_1$-minimization. The results suggest that both the primal scheme and a slightly revised accelerated scheme converge at a geometric rate, that is, much faster than what our theoretical analysis guarantees. This suggests the open problem of improving the quality of our error bounds. 

To compare our approaches to other algorithms for $\ell_1$-minimization, we implemented them in MATLAB and ran the \texttt{l1benchmark} suite by Yang et al.~\cite{Yang:2013}, which includes implementations of many other $\ell_1$-minimization solvers. 
A representative comparison is shown in Figure \ref{fig:exp1}. 
\begin{figure}
\begin{center}
\subfloat{\includegraphics[scale=0.5]{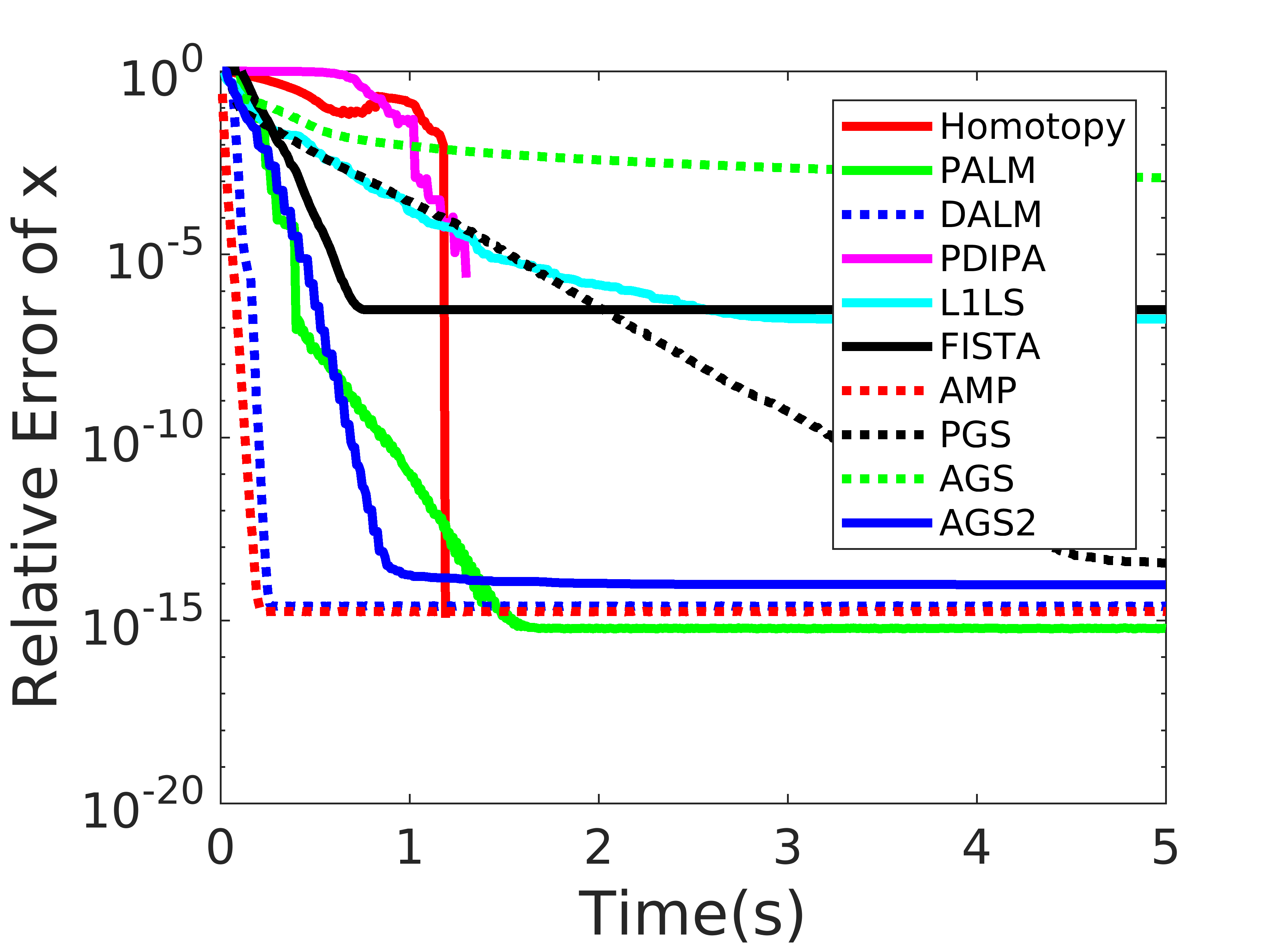}}
\subfloat{\includegraphics[scale=0.5]{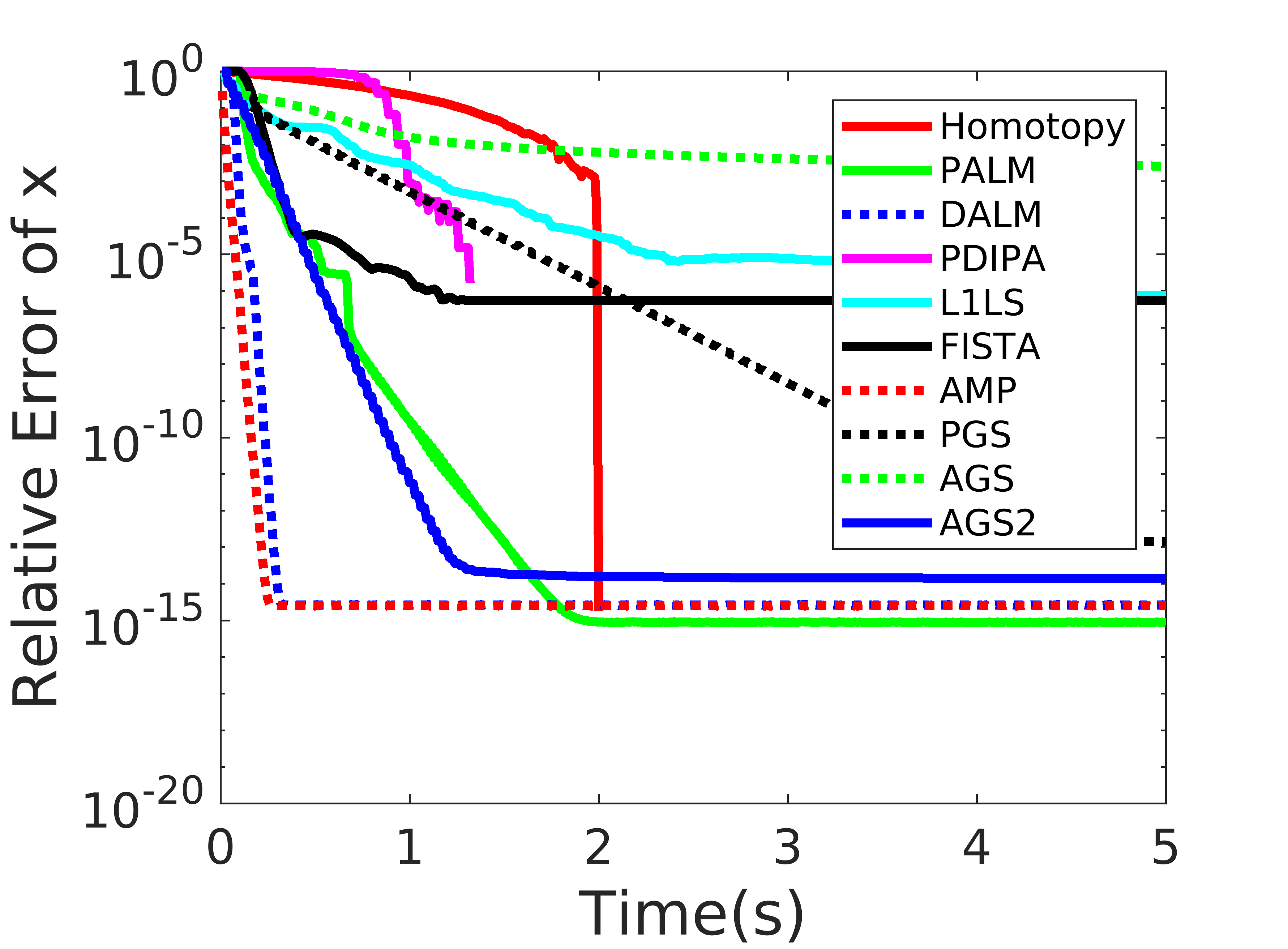}}
\caption{Results from the {\tt l1benchmark} package. Left: $m=1000$, $n=800$, $20\%$ density; right: $m=1000$, $n=800$, $30\%$ density}
\label{fig:exp1}
\end{center}
\end{figure}
The figure plots the relative error of the algorithms as a function of computation time, averaged on 20 randomly generated instances (with $m=1000$, $n=800$, and 20\% or 30\% nonzeros in the ground truth solution). The implementations based on our approaches are:
\begin{itemize}
\item the Primal Gradient Scheme of Section \ref{sec:pgs} (PGS, with $\beta=3.5$, $\delta=10^{-15}$),
\item the Accelerated Gradient Scheme of Section \ref{sec:ags} (AGS, with $\beta=3.5$, $\delta=10^{-15}$, $\tau_k=2/(k+3)$), and 
\item a revised Accelerated Gradient Scheme, which we formulate below (AGS2, with $\beta=1.1$, $\delta=10^{-15}$, $\tau_k=10^{-15}$). 
\end{itemize} 
Other algorithms measured in the experiment are the Homotopy method, the primal and dual augmented Lagrangian methods (PALM, DALM), the primal-dual interior point method (PDIPA), the truncated Newton interior point method (L1LS), the fast iterative soft-thresholding method (FISTA), and the approximate message passing method (AMP). We refer the reader to Yang et al.~\cite{Yang:2013} for references and discussion of these other methods. We observe, incidentally, that many of these methods construct points that are only approximately feasible for (BP), since they relax the constraint $As=b$ into the objective function, in one form or the other. 


In the experiments, 
PGS clearly exhibits a geometric convergence rate, which is much better than what Theorem \ref{thm:pgs-for-dm} guarantees, strongly suggesting that an improved theoretical analysis may be possible. Over time, PGS essentially reaches the machine precision barrier ($ \approx 10^{-15}$), which is not true for other methods in the benchmark, such as FISTA, L1LS or the interior point method (PDIPA). 

AGS, on the other hand, appears to be rather inaccurate in practice and does not exhibit a substantially better behavior than what is guaranteed by Theorem \ref{thm:ags-for-dm}.  
This suggests that the entropic form of the updates -- used in PGS but not in AGS -- might have a high impact in practice. Therefore, we also benchmark a revised algorithm (AGS2) obtained by adopting an entropic form of the AGS updates \eqref{eq:explicit-ags-1}--\eqref{eq:explicit-ags-2}, as follows (colored terms are new): 
\begin{align}
y^k_j &= \max \{ \delta, x^k_j - {\color{blue} x^k_j} \cdot \beta^{-1} [\nabla f(x^k)]_j \} \\
z^k_j &= \max \{ \delta, x^0_j - {\color{blue} x^0_j} \cdot \beta^{-1} [\sum_i \alpha_i \nabla f(x^i)]_j \} \\
x^{k+1} &= \tau_k z^k + (1-\tau_k) y^k. 
\end{align}
The resulting scheme AGS2 is seen in Figure \ref{fig:exp1} to exhibit a geometric convergence rate and to be competitive against some of the best results in the benchmark, such as those of the primal augmented Lagrangian method (PALM). 

\section{Conclusions}
\label{sec:concl}
We proposed a novel exact reformulation of the basis pursuit problem, which leads to a new family of gradient-based, IRLS-type methods for its solution. We then analyzed the iteration complexity of a natural optimization approach to the reformulation, based on the mirror descent scheme, as well as the iteration complexity of an accelerated gradient method. 
The first scheme can be seen as the discretization of a Hessian gradient flow and also as a variant on the Physarum dynamics, derived purely from optimization principles. The accelerated method, on the other hand, improves the error dependency for IRLS-type methods for basis pursuit, from $\eps^{-8/3}$ to $\eps^{-2}$. The experimental convergence rate of the first scheme, as well as that of a simple variant the second scheme, appears to be geometric. We interpret this as evidence that the dissipation minimization perspective may stimulate even more approaches to the design and analysis of efficient and practical IRLS-type methods.


\begin{thebibliography}{10}

\bibitem{Alvarez:2004}
F.~Alvarez, J.~Bolte, and O.~Brahic.
\newblock Hessian {Riemannian} gradient flows in convex programming.
\newblock {\em SIAM J. Control and Optimization}, 43(2):477--501, 2004.

\bibitem{Amari:2016}
S.~Amari.
\newblock {\em Information Geometry and Its Applications}.
\newblock Springer, 2016.

\bibitem{Arora:2012}
S.~Arora, E.~Hazan, and S.~Kale.
\newblock The multiplicative weights update method: a meta-algorithm and
  applications.
\newblock {\em Theory of Computing}, 8(1):121--164, 2012.

\bibitem{Bach:2012}
{F.~R.} Bach, R.~Jenatton, J.~Mairal, and G.~Obozinski.
\newblock Optimization with sparsity-inducing penalties.
\newblock {\em Foundations and Trends in Machine Learning}, 4(1):1--106, 2012.

\bibitem{Bansal:2019}
N.~Bansal and A.~Gupta.
\newblock Potential-function proofs for gradient methods.
\newblock {\em Theory of Computing}, 15(4):1--32, 2019.

\bibitem{Bauschke:2017}
{H.~H.} Bauschke, J.~Bolte, and M.~Teboulle.
\newblock A descent lemma beyond {Lipschitz} gradient continuity: First-order
  methods revisited and applications.
\newblock {\em Math. Oper. Res.}, 42(2):330--348, 2017.

\bibitem{Beck:2015}
A.~Beck.
\newblock On the convergence of alternating minimization for convex programming
  with applications to iteratively reweighted least squares and decomposition
  schemes.
\newblock {\em SIAM Journal on Optimization}, 25(1):185--209, 2015.

\bibitem{Beck:2017}
A.~Beck.
\newblock {\em First-Order Methods in Optimization}.
\newblock SIAM, 2017.

\bibitem{Beck:2003}
A.~Beck and M.~Teboulle.
\newblock Mirror descent and nonlinear projected subgradient methods for convex
  optimization.
\newblock {\em Oper. Res. Lett.}, 31(3):167--175, 2003.

\bibitem{Becker:2017}
R.~Becker, V.~Bonifaci, A.~Karrenbauer, P.~Kolev, and K.~Mehlhorn.
\newblock Two results on slime mold computations.
\newblock {\em Theoretical Computer Science}, 773:79--106, 2019.

\bibitem{Bloomfield:1983}
P.~Bloomfield and W.~L. Steiger.
\newblock {\em Least Absolute Deviations: Theory, Applications, and
  Algorithms}.
\newblock Birkh{\"a}user, 1983.

\bibitem{Bollobas:1998}
B.~Bollob{\'a}s.
\newblock {\em Modern Graph Theory}.
\newblock Springer, New York, 1998.

\bibitem{Bonifaci:2017:b}
V.~Bonifaci.
\newblock On the convergence time of a natural dynamics for linear programming.
\newblock In {\em Proc.~of the 28th Int.~Symposium on Algorithms and
  Computation}, pages 17:1--17:12. Schloss Dagstuhl--Leibniz-Zentrum fuer
  Informatik, 2017.

\bibitem{Bonifaci:2012}
V.~Bonifaci, K.~Mehlhorn, and G.~Varma.
\newblock Physarum can compute shortest paths.
\newblock In {\em Proc. of the 23rd ACM-SIAM Symposium on Discrete Algorithms},
  pages 233--240. SIAM, 2012.

\bibitem{Boyd:2004}
S.~Boyd and L.~Vanderberghe.
\newblock {\em Convex Optimization}.
\newblock Cambridge University Press, 2004.

\bibitem{Candes:2005}
E.~Cand{\`e}s and J.~Romberg.
\newblock $\ell_1$-magic: Recovery of sparse signals via linear programming.
\newblock
  \url{https://statweb.stanford.edu/~candes/l1magic/downloads/l1magic.pdf},
  2005.

\bibitem{Chartrand:2008}
R.~Chartrand and W.~Yin.
\newblock Iteratively reweighted algorithms for compressive sensing.
\newblock In {\em Proc.~of IEEE Int. Conf. on Acoustics, Speech and Signal
  Processing}, pages 3869--3872. IEEE, 2008.

\bibitem{Chen:2001}
S.~S. Chen, D.~L. Donoho, and M.~A. Saunders.
\newblock Atomic decomposition by basis pursuit.
\newblock {\em {SIAM} Review}, 43(1):129--159, 2001.

\bibitem{Chin:2013}
H.~H. Chin, A.~Madry, G.~L. Miller, and R.~Peng.
\newblock Runtime guarantees for regression problems.
\newblock In {\em Proc. of Innovations in Theoretical Computer Science}, pages
  269--282. ACM, 2013.

\bibitem{Christiano:2011}
P.~Christiano, J.~A. Kelner, A.~Madry, D.~A. Spielman, and S.-H. Teng.
\newblock Electrical flows, {Laplacian} systems, and faster approximation of
  maximum flow in undirected graphs.
\newblock In {\em Proc. of the 43rd ACM Symp. on Theory of Computing}, pages
  273--282. ACM, 2011.

\bibitem{Daubechies:2010}
I.~Daubechies, R.~DeVore, M.~Fornasier, and C.S. G{\"u}nt{\"u}rk.
\newblock Iteratively reweighted least squares minimization for sparse
  recovery.
\newblock {\em Comm. on Pure and Applied Mathematics}, 63(1):1--38, 2010.

\bibitem{Ene:2019}
A.~Ene and A.~Vladu.
\newblock Improved convergence for $\ell_1$ and $\ell_\infty$ regression via
  iteratively reweighted least squares.
\newblock In {\em Proceedings of the 36th International Conference on Machine
  Learning}, pages 1794--1801, 2019.

\bibitem{Facca:2019}
E.~Facca, F.~Cardin, and M.~Putti.
\newblock Physarum dynamics and optimal transport for basis pursuit.
\newblock {\tt arXiv:1812.11782v1 [math.NA]}, 2019.

\bibitem{Foucart:2013}
S.~Foucart and H.~Rauhut.
\newblock {\em A Mathematical Introduction to Compressive Sensing}.
\newblock Birkh{\"{a}}user, 2013.

\bibitem{Ghosh:2008}
A.~Ghosh, S.~Boyd, and A.~Saberi.
\newblock Minimizing effective resistance of a graph.
\newblock {\em {SIAM} Review}, 50(1):37--66, 2008.

\bibitem{Green:1984}
P.~J. Green.
\newblock Iteratively reweighted least squares for maximum likelihood
  estimation, and some robust and resistant alternatives.
\newblock {\em Journal of the Royal Statistical Society, Series B},
  46(2):149--192, 1984.

\bibitem{Hofbauer:1998}
J.~Hofbauer and K.~Sigmund.
\newblock {\em Evolutionary Games and Population Dynamics}.
\newblock Cambridge University Press, 1998.

\bibitem{Horn:2013}
R.~A. Horn and C.~R. Johnson.
\newblock {\em Matrix Analysis}.
\newblock Cambridge University Press, 2013.

\bibitem{Lu:2018}
H.~Lu, R.~M. Freund, and Yu. Nesterov.
\newblock Relatively smooth convex optimization by first-order methods, and
  applications.
\newblock {\em {SIAM} Journal on Optimization}, 28(1):333--354, 2018.

\bibitem{Nemirovski:1983}
A.~S. Nemirovski and D.~B. Yudin.
\newblock {\em Problem Complexity and Method Efficiency in Optimization}.
\newblock John Wiley, 1983.

\bibitem{Nesterov:2005}
Yu. Nesterov.
\newblock Smooth minimization of non-smooth functions.
\newblock {\em Math. Program.}, 103(1):127--152, 2005.

\bibitem{Osborne:1985}
M.~R. Osborne.
\newblock {\em Finite Algorithms in Optimization and Data Analysis}.
\newblock Wiley, 1985.

\bibitem{Rockafellar:1970}
R.~T. Rockafellar.
\newblock {\em Convex Analysis}.
\newblock Princeton University Press, 1970.

\bibitem{Strang:1988}
G.~Strang.
\newblock A framework for equilibrium equations.
\newblock {\em SIAM Review}, 30(2):283--296, 1988.

\bibitem{Straszak:2016:irls}
D.~Straszak and N.~K. Vishnoi.
\newblock {IRLS} and slime mold: Equivalence and convergence.
\newblock {\tt arXiv:1601.02712 [cs.DS]}, 2016.
\newblock \href {http://arxiv.org/abs/1601.02712} {\path{arXiv:1601.02712}}.

\bibitem{Straszak:2016:soda}
D.~Straszak and N.~K. Vishnoi.
\newblock Natural algorithms for flow problems.
\newblock In {\em Proc. of the 27th {ACM-SIAM} Symposium on Discrete
  Algorithms}, pages 1868--1883. SIAM, 2016.

\bibitem{Tero:2007}
A.~Tero, R.~Kobayashi, and T.~Nakagaki.
\newblock A mathematical model for adaptive transport network in path finding
  by true slime mold.
\newblock {\em Journal of Theoretical Biology}, 244:553--564, 2007.

\bibitem{Wilson:2018}
A.~Wilson.
\newblock Lyapunov arguments in optimization.
\newblock Ph.D.~dissertation, University of California at Berkeley, 2018.

\bibitem{Yang:2013}
A.~Y. Yang, Z.~Zhou, A.~G. Balasubramanian, S.~S. Sastry, and Y.~Ma.
\newblock Fast $\ell_1$-minimization algorithms for robust face recognition.
\newblock {\em {IEEE} Trans. Image Processing}, 22(8):3234--3246, 2013.

\end{thebibliography}

%
%
%
%





\end{document}